\documentclass[preprint,12pt]{elsarticle}

\usepackage[T1]{fontenc}
\IfFileExists{lmodern.sty}{\usepackage{lmodern}}{}
\usepackage{amsmath,amssymb,amsthm}
\usepackage{booktabs}
\usepackage{enumitem}
\usepackage[hidelinks]{hyperref}

\newtheorem{theorem}{Theorem}
\newtheorem{corollary}[theorem]{Corollary}
\newtheorem{proposition}[theorem]{Proposition}
\newtheorem{lemma}[theorem]{Lemma}
\theoremstyle{definition}
\newtheorem{remark}[theorem]{Remark}

\newcommand{\OPT}{\mathrm{OPT}}
\newcommand{\greedy}{\mathrm{greedy}}
\newcommand{\Sbar}{\overline{S}}

\journal{Operations Research Letters}

\begin{document}

\begin{frontmatter}

\title{When is greedy exact for weighted minimum 3-cut?}

\author{Armando Rodado Amaris}
\ead{armando@rodadolab.com}
\affiliation{organization={Centre for Quantum Software and Information, University of Technology Sydney},
  city={Sydney},
  state={NSW},
  country={Australia}}

\begin{abstract}
The greedy algorithm for minimum 3-cut splits a global minimum-weight cut and then the cheaper of its
two sides; it is a tight $4/3$-approximation. We show that it is exact, from every global
minimum-weight cut, whenever some optimal 3-partition has two parts joined by total weight at least
the weighted edge connectivity $\lambda$. This extends both theorems of Hochbaum and Shmoys from
unweighted planar graphs to all weighted graphs, recovers the $4/3$ bound, and generalizes to minimum
$k$-cut with a one-paragraph proof. The remaining hard instances have all pair weights below
$\lambda$; for them we prove a merge lemma for Gomory--Hu sides. Finally, we give a 7-vertex graph
with a unique Gomory--Hu tree whose unique optimal 3-partition contains no side of that tree, so the
first stage of Kapoor's algorithm does not suffice on its own. The main theorem, the $4/3$ bound,
the merge lemmas and the counterexample are verified in the Lean~4 proof assistant.
\end{abstract}


\end{frontmatter}

\section{Introduction}

Given an undirected graph $G=(V,E)$ with positive edge weights, a \emph{3-cut} is a set of edges whose
removal leaves at least three components; equivalently, the set of edges crossing a partition of $V$
into three nonempty parts. The minimum 3-cut problem asks for one of least total weight. Throughout,
``minimum'' refers to total weight, not to the number of edges (see Section~\ref{sec:prelim}).

The problem is polynomial for fixed $k$~\cite{GoldschmidtHochbaum1994}. Exact algorithms for $k=3$ have
successively used $O(n^3)$ max-flow computations~\cite{Kapoor1996,KamidoiWY2002}, an $O(mn^3)$
enumeration of light cuts~\cite{BurletGoldschmidt1997}, and $O(n)$ candidate cuts drawn from an ordered
cut enumeration~\cite{NagamochiIbaraki2000}. Vaughn~\cite{Vaughn2026} recently obtained
$O(n^2\log^2 n)$ randomized and $\tilde O(n^2)$ deterministic time; under the Max-Weight Clique
conjecture this is essentially optimal in general weighted graphs.

The simplest heuristic is the \emph{greedy algorithm}: take a global minimum-weight cut $(S,V\setminus S)$,
then a minimum-weight cut of whichever side gives the lighter result. For weighted graphs it is a tight
$4/3$-approximation~\cite{HochbaumTsai1983}, as reported in~\cite{HochbaumShmoys1985}. For unweighted
planar graphs, Hochbaum and Shmoys~\cite{HochbaumShmoys1985} proved, in the language of dual cycles, that it
is exact in two structural situations. This note asks when the greedy algorithm is exact in general.

\paragraph{Contributions.}
\begin{enumerate}[leftmargin=*]
\item \emph{Exactness criterion} (Theorem~\ref{thm:main}). For every 3-partition $(A,B,C)$ and every
  global minimum-weight cut $S$, $\greedy(S)\le\max(\lambda+d(A),\,W(A,B,C))$. Consequently greedy is
  exact whenever some optimum has a pair of parts joined by weight at least $\lambda$. This holds in
  every weighted graph, recovers both Hochbaum--Shmoys theorems (Corollary~\ref{cor:hs}) and the $4/3$
  bound (Corollary~\ref{cor:43}), and gives near-linear exact algorithms outside a well-defined hard
  class on general, planar and bounded-genus graphs (Corollary~\ref{cor:time}). Theorem~\ref{thm:main},
  Corollaries~\ref{cor:hs} and~\ref{cor:43}, Proposition~\ref{prop:hard} and the merge lemmas are
  formalized in Lean~4.
\item \emph{Hard instances} (Section~\ref{sec:hard}). When the criterion fails, all three pair weights
  are below $\lambda$, forcing every part's cut weight into $[\lambda,2\lambda)$ and the optimum into
  $[3\lambda/2,3\lambda)$. A merge lemma shows that any Gomory--Hu side missing a vertex of each of the
  other two parts can be absorbed into a part of an optimum.
\item \emph{Gomory--Hu sides are not enough} (Section~\ref{sec:cex}). Kapoor's
  algorithm~\cite{Kapoor1996} first tries every side $X$ of a Gomory--Hu tree as a part, completing it by
  a minimum-weight cut of $G[V\setminus X]$, and then needs further recursion. We show this cannot be
  avoided: a 7-vertex graph has a unique Gomory--Hu tree and a unique minimum-weight 3-partition
  containing no Gomory--Hu side. All claims about this graph are verified in Lean~4.
\item \emph{Minimum $k$-cut} (Section~\ref{sec:kcut}). A $k$-partition is \emph{compatible} with a
  cut $S$ if each of its parts lies inside $S$ or inside $V\setminus S$. For every $k\ge3$, every
  global minimum-weight cut $S$ and every $k$-partition $P$ with parts $A_i,A_j$, the cheapest
  compatible $k$-partition weighs at most $\max(W(P)-w(A_i,A_j)+\lambda,\,W(P))$
  (Theorem~\ref{thm:kcut}). Hence every global minimum-weight cut is compatible with some
  minimum-weight $k$-partition as soon as some optimum has two parts joined by weight at least
  $\lambda$. Without this condition a global minimum-weight cut can be incompatible with every
  optimum, already for $k=3$, where this is exactly the failure of greedy. The theorem also gives a
  two-line proof of Theorem~\ref{thm:main}.
\end{enumerate}

\paragraph{Relation to prior arguments.}
Two cases in the proof of Theorem~\ref{thm:main}, a minimum-weight cut crossing the lightest part and
the lightest part contained in one side, overlap Lemmas~2 and~3 of Kamidoi, Wakabayashi and
Yoshida~\cite{KamidoiWY2002}. The case where Vaughn's algorithm finishes immediately (a part of weight
below $\OPT/2$) is a special case of our criterion. The merge lemma is related to Xiao's non-crossing
lemmas for 3-way cuts in hypergraphs~\cite{Xiao2008}.

\section{Preliminaries}\label{sec:prelim}

Throughout, $G=(V,E)$ is connected with positive edge weights $w$ and $n\ge 3$ vertices. For
$X\subseteq V$, $d(X)$ is the weight of the edges with exactly one end in $X$, and
$\lambda=\min\{d(X):\emptyset\ne X\subsetneq V\}$ is the \emph{weighted edge connectivity}, the least
total weight of a cut. For disjoint $P,Q$, $w(P,Q)$ is the weight of the edges joining them. For a
3-partition $(A,B,C)$ write $w_{AB}=w(A,B)$ and so on, and
\[
  W(A,B,C) \;=\; w_{AB}+w_{AC}+w_{BC} \;=\; \tfrac12\bigl(d(A)+d(B)+d(C)\bigr).
\]
$\OPT$ is the minimum of $W$ over 3-partitions; a minimum-weight 3-partition is an \emph{optimum}. For
$U\subseteq V$ with $|U|\ge2$, $\lambda(G[U])$ is the least weight of a bipartition of the induced
subgraph into two nonempty sets ($0$ if $G[U]$ is disconnected); set $\lambda(G[U])=\infty$ if $|U|=1$.

\paragraph{Terminology.} All graphs are weighted, and ``minimum'' always refers to total edge weight,
never to the number of edges. A \emph{global minimum-weight cut} is a set $S$ with $d(S)=\lambda$; in
general it need not have the fewest edges among all cuts, and a cut with the fewest edges need not
have weight $\lambda$. The two notions agree only for unit weights, which is the setting of
Hochbaum and Shmoys~\cite{HochbaumShmoys1985}. Likewise a minimum-weight $x$--$y$ cut is a cut
separating $x$ from $y$ of least total weight $\lambda(x,y)$, and a minimum-weight 3-cut is a
3-partition of least total crossing weight.

The cut function is symmetric and submodular, hence posimodular: for all $X,Y\subseteq V$,
\[
  d(X)+d(Y)\ge d(X\cap Y)+d(X\cup Y), \qquad d(X)+d(Y)\ge d(X\setminus Y)+d(Y\setminus X).
\]

For a global minimum-weight cut $S$ (so $d(S)=\lambda$), let
\[
  \greedy(S) \;=\; \lambda + \min\bigl(\lambda(G[S]),\,\lambda(G[V\setminus S])\bigr).
\]
Splitting one side of $S$ optimally gives a 3-partition of this weight, so $\greedy(S)\ge\OPT$; greedy
is \emph{exact} if $\greedy(S)=\OPT$.

A \emph{Gomory--Hu tree}~\cite{GomoryHu1961} $T$ is a weighted tree on $V$ such that for every tree
edge $e=(x,y)$ the side $F_e$ of $T-e$ containing $x$ satisfies $d(F_e)=w_T(e)=\lambda(x,y)$, the
minimum weight of a cut separating $x$ from $y$. We call $F_e$ and $V\setminus F_e$ the
\emph{Gomory--Hu sides} of $e$. Every $Y\subseteq F_e$ with $x\in Y$ separates $x$ from $y$, so
$d(Y)\ge d(F_e)$.

\section{Greedy exactness}

\begin{theorem}\label{thm:main}
For every 3-partition $(A,B,C)$ of $V$ and every global minimum-weight cut $S$,
\[
  \greedy(S)\;\le\;\max\bigl(\lambda+d(A),\;W(A,B,C)\bigr).
\]
In particular, if some optimum $(A,B,C)$ satisfies $w_{BC}\ge\lambda$, then $\greedy(S)=\OPT$ for
every global minimum-weight cut $S$.
\end{theorem}

\begin{proof}
Write $\Sbar=V\setminus S$. It suffices to exhibit a bipartition of $S$ or of $\Sbar$ into two nonempty
sets whose crossing weight is at most $d(A)$; adding $\lambda=d(S)$ then bounds $\greedy(S)$ by
$\lambda+d(A)$. The one exception is case~(ii), which yields $W(A,B,C)$ directly. Since $A\ne V$,
exactly one of the following holds, up to exchanging $S$ and $\Sbar$.
\begin{enumerate}[label=(\roman*),leftmargin=*]
\item $A\subsetneq S$. Split $S$ into $A$ and $S\setminus A$; the crossing weight is at most $d(A)$.
\item $A=S$. Split $\Sbar$ into $B$ and $C$: $\greedy(S)\le\lambda+w_{BC}=d(A)+w_{BC}=W(A,B,C)$.
\item $S\subsetneq A$. Then $B\cup C\subseteq\Sbar$ and $A\setminus S\ne\emptyset$. Split $\Sbar$ into
  $A\setminus S$ and $B\cup C$; the crossing weight is
  $w(A\setminus S,B\cup C)\le w(A,B\cup C)=d(A)$.
\item $A$ and $S$ cross, i.e.\ $A\cap S$, $A\setminus S$, $S\setminus A$ and $V\setminus(A\cup S)$ are
  all nonempty. By posimodularity, $d(A)+d(S)\ge d(A\setminus S)+d(S\setminus A)$. Since
  $S\setminus A$ is a nonempty proper subset, $d(S\setminus A)\ge\lambda=d(S)$, so
  $d(A\setminus S)\le d(A)$. Split $\Sbar$ into $A\setminus S$ and $\Sbar\setminus A$, both nonempty;
  the crossing weight is at most $d(A\setminus S)\le d(A)$.
\end{enumerate}
For the second statement, $d(A)=w_{AB}+w_{AC}$, so $\lambda+d(A)\le W(A,B,C)$ exactly when
$\lambda\le w_{BC}$, and then $\greedy(S)\le\OPT\le\greedy(S)$.
\end{proof}

For the lightest part, cases (i) and (iv) essentially appear as Lemmas~3 and~2 of~\cite{KamidoiWY2002}.
The new ingredients are case~(iii) and the uniform pair-weight condition.
Theorem~\ref{thm:main} also follows from the $k$-cut argument of Section~\ref{sec:kcut} applied to
the pair $(B,C)$, since $W(A,B,C)-w_{BC}=d(A)$.

\begin{corollary}[weighted Hochbaum--Shmoys]\label{cor:hs}
Let $S$ be any global minimum-weight cut.
\begin{enumerate}[label=(\alph*),leftmargin=*]
\item If some optimum has two parts $A,B$ joined by no edge, then $\greedy(S)=\OPT$.
\item If some optimum has a pair of parts joined by weight at least $\lambda$, then $\greedy(S)=\OPT$.
\end{enumerate}
\end{corollary}

\begin{proof}
In (a), $w_{AC}=d(A)\ge\lambda$, so Theorem~\ref{thm:main} applies with the parts permuted; (b) is
Theorem~\ref{thm:main}.
\end{proof}

For unweighted planar graphs, (a) and (b) are Theorems~1 and~2 of~\cite{HochbaumShmoys1985}, stated
there for dual cycles; a theta-shaped dual with a strand at least as long as the girth corresponds
to~(b).

\begin{corollary}[the $4/3$ bound]\label{cor:43}
$\greedy(S)\le\tfrac43\,\OPT$ for every global minimum-weight cut $S$.
\end{corollary}

\begin{proof}
Choose $A$ opposite the heaviest pair, so $w_{BC}\ge\OPT/3$. Theorem~\ref{thm:main} gives
$\greedy(S)\le\max(\OPT+\lambda-w_{BC},\,\OPT)$, and
$\lambda\le\min(d(A),d(B),d(C))\le\tfrac23\OPT$.
\end{proof}

This recovers the tight bound of~\cite{HochbaumTsai1983}.

\begin{corollary}[running time]\label{cor:time}
The greedy algorithm needs one global minimum-weight cut of $G$ and one of each side; the sides are
edge-disjoint and inherit planarity or an embedding of genus $g$. It returns an optimum whenever some
optimum satisfies the criterion of Theorem~\ref{thm:main}, within the following times.
\begin{center}
\begin{tabular}{lll}
\toprule
Graph class & Global minimum-weight cut subroutine & Time\\
\midrule
General weighted & Henzinger--Li--Rao--Wang~\cite{HenzingerLRW2024} (deterministic) & near-linear\\
Planar & {\L}{\k{a}}cki--Sankowski~\cite{LackiSankowski2011} & $O(n\log\log n)$\\
Genus $g$ & Chambers--Erickson--Fox--Nayyeri~\cite{ChambersEFN} & $g^{O(g)}\,n\log\log n$\\
\bottomrule
\end{tabular}
\end{center}
\end{corollary}

\section{Hard instances}\label{sec:hard}

Call an optimum $(A,B,C)$ \emph{hard} if all three pair weights are below $\lambda$. By
Corollary~\ref{cor:hs}, greedy is exact unless every optimum is hard; so $\min(\greedy,H)$ is exact for
any procedure $H$ that always returns a valid 3-cut and is exact when all optima are hard.

\begin{proposition}\label{prop:hard}
Let $(A,B,C)$ be a hard optimum. Then
(1) $0<w_{XY}<\lambda$ for every pair of parts;
(2) $\lambda\le d(X)<2\lambda$ for every part $X$;
(3) $\tfrac32\lambda\le\OPT<3\lambda$.
\end{proposition}

\begin{proof}
Each $d(X)\ge\lambda$ is a sum of two pair weights, each below $\lambda$, so both are positive and
$d(X)<2\lambda$. Summing, $2\,\OPT=\sum_X d(X)\ge3\lambda$, while $\OPT$, the sum of the three pair
weights, is below $3\lambda$.
\end{proof}

Every part of a hard optimum is thus a 2-approximate minimum-weight cut. In a plane graph, the dual of
a hard optimal cut is a theta-graph whose three strands are each lighter than the girth of the dual;
the other two dual shapes, two disjoint cycles or two cycles sharing a vertex, leave two parts with no
edges between them.

The next two lemmas enlarge a part of an optimum by a suitable minimum-weight cut without losing
optimality. Lemma~\ref{lem:merge-cut} is closely related to Xiao's Lemma~6~\cite{Xiao2008}, which gives
the corresponding non-crossing statement; the explicit merged optimum is what we use.

\begin{lemma}\label{lem:merge-cut}
Let $(A,B,C)$ be an optimum, $a\in A$, and $R\subseteq V\setminus\{a\}$ a set meeting both $B$ and $C$.
If $X$ is the $a$-side of any minimum-weight cut separating $a$ from $R$, then
$(A\cup X,\,B\setminus X,\,C\setminus X)$ is an optimum.
\end{lemma}

\begin{proof}
The sets $A\cap X$, $X\setminus B$ and $X\setminus C$ each contain $a$ and avoid $R$, so each has cut
weight at least $d(X)$. Submodularity and posimodularity give
\begin{align*}
d(A)+d(X)&\ge d(A\cup X)+d(A\cap X),\\
d(B)+d(X)&\ge d(B\setminus X)+d(X\setminus B),\\
d(C)+d(X)&\ge d(C\setminus X)+d(X\setminus C).
\end{align*}
Summing, $2\,\OPT+3d(X)\ge 2\,W(A\cup X,B\setminus X,C\setminus X)+3d(X)$. The new parts are nonempty
because $R$ meets $B$ and $C$.
\end{proof}

\begin{lemma}\label{lem:merge-gh}
Let $(A,B,C)$ be an optimum and $T$ a Gomory--Hu tree. Let $e=(s,t)$ be a tree edge with $s\in A$, and
let $X$ be its Gomory--Hu side containing $s$. If $X$ contains neither all of $B$ nor all of $C$, then
$(A\cup X,\,B\setminus X,\,C\setminus X)$ is an optimum.
\end{lemma}

\begin{proof}
Every $Y\subseteq X$ with $s\in Y$ separates $s$ from $t$, so $d(Y)\ge d(X)$. The sets $A\cap X$,
$X\setminus B$ and $X\setminus C$ contain $s$ and lie in $X$, so the three inequalities of
Lemma~\ref{lem:merge-cut} apply. The new parts are nonempty by hypothesis.
\end{proof}

\begin{remark}
If such an $X$ has at least two vertices, contracting $X$ preserves $\OPT$, since some optimum keeps
$X$ inside one part, and $T$ with $X$ contracted is a Gomory--Hu tree of $G/X$. Any statement of the
form ``some optimum has a part equal to a Gomory--Hu side'' therefore reduces, by induction on $n$, to
optima admitting no such merge. Section~\ref{sec:cex} shows that this residual case genuinely occurs.
\end{remark}

\section{Gomory--Hu sides do not suffice}\label{sec:cex}

Given a Gomory--Hu tree $T$, Kapoor's algorithm~\cite{Kapoor1996} first evaluates
\[
  F(T) \;=\; \min_{X}\Bigl(d(X)+\lambda\bigl(G[V\setminus X]\bigr)\Bigr),
\]
where $X$ ranges over the Gomory--Hu sides of $T$. It then recursively shrinks sides until $T$ is a
star, and handles that residual case with $O(n^3)$ max-flow computations. $F(T)$ is always the weight
of a 3-cut, so $F(T)\ge\OPT$, with equality whenever some optimum has a part that is a Gomory--Hu side.
Equality also holds whenever greedy is exact, in particular under the criterion of
Theorem~\ref{thm:main}: the side $S$ of a lightest tree edge is a global minimum-weight cut, and
$F(T)\le\greedy(S)$. One might hope that equality always holds, which would make
the recursion unnecessary. It does not.

\begin{theorem}\label{thm:cex}
There is a weighted graph on 7 vertices with a unique Gomory--Hu tree $T$ and a unique optimum, such
that no part of the optimum is a side of $T$ and $F(T)>\OPT$.
\end{theorem}

\paragraph{The instance.}
Vertices $s,u_0,\ell_0,\ell_1,\ell_2,c_0,c_1$ and edge weights:
\begin{center}
\begin{tabular}{lr@{\qquad}lr}
\toprule
Edge & Weight & Edge & Weight\\
\midrule
$c_0c_1$ & 0.337164 & $\ell_2\ell_0$ & 0.206713\\
$c_0\ell_0$ & 0.166989 & $sc_0$ & 0.035930\\
$c_1\ell_0$ & 0.035652 & $s\ell_1$ & 0.210285\\
$\ell_1\ell_0$ & 0.221789 & $s\ell_2$ & 0.111054\\
$\ell_1\ell_2$ & 0.194170 & $su_0$ & 0.364863\\
$\ell_2c_0$ & 0.110311 & $u_0c_0$ & 0.035173\\
 & & $u_0c_1$ & 0.327062\\
\bottomrule
\end{tabular}
\end{center}

\paragraph{Verification.}
All quantities below are for the weights exactly as printed. They were computed by exhaustive
enumeration in exact integer arithmetic on the weights $10^6w$: cut values over all 126 nonempty
proper vertex subsets, the optimum over all 3-partitions, and the Gomory--Hu tree over all
$7^5=16{,}807$ labelled spanning trees. Every item has also been verified in the Lean~4 proof
assistant, by kernel evaluation of the corresponding finite statements (no compiled code is trusted);
there, uniqueness of the tree is checked over all $7^6$ parent arrays rooted at $s$, which include
every spanning tree. Both verifications are provided as supplementary material.
\begin{itemize}[leftmargin=*]
\item \emph{The tree.} $T$ is the caterpillar with edges $s\ell_0$ (0.631143), $s\ell_1$ (0.626244),
  $s\ell_2$ (0.622248), $su_0$ (0.713745), $u_0c_0$ (0.685567) and $u_0c_1$ (0.699878). It is the only
  spanning tree with the cut-tree property, and the minimum-weight cut of each of its edges is unique.
\item \emph{The optimum.} It is unique: $(\{s,u_0\},\{\ell_0,\ell_1,\ell_2\},\{c_0,c_1\})$, with
  $\OPT=1.032456$.
\item \emph{Hardness.} $\lambda=0.622248$ and the pair weights are $0.312952$, $0.398165$ and
  $0.321339$, so the optimum is hard and greedy is not exact.
\item \emph{The gap.} The sides of $T$ are the five single leaves, their complements,
  $\{s,\ell_0,\ell_1,\ell_2\}$ and its complement $\{u_0,c_0,c_1\}$. Completing each optimally gives
  $F(T)=1.035084>\OPT$.
\end{itemize}
The instance is a random $5\%$ perturbation of a vertex of a linear program, so the failure is not an
artefact of ties.

\paragraph{How the instance was found.}
By Lemma~\ref{lem:merge-gh}, $F(T)$ can exceed $\OPT$ only in a graph none of whose optima admits a
merge of a Gomory--Hu side with at least two vertices. This restricts how the parts can sit in $T$.
Guided by it, we searched caterpillar trees: a spine of vertices of a part $A$, with a part $B$ as the
leaves at one end of the spine and a part $C$ as the leaves at the other. For a fixed tree shape, the condition that $T$ be a
Gomory--Hu tree is linear in the edge weights: every cut separating the endpoints of a tree edge weighs
at least that edge's fundamental cut. A linear program then maximizes the margin by which every
3-partition with a Gomory--Hu-side part loses to $(A,B,C)$. The margin is $0$ for $|B|=1$, for
$|B|=2$ with $|C|\in\{1,2\}$, and for $(|B|,|C|)=(3,1)$; it is positive for $(|B|,|C|)=(3,2)$, which is
the instance above.

\section{Minimum $k$-cut}\label{sec:kcut}

For $k\ge2$, $W(P)$ denotes the weight of the edges crossing a $k$-partition $P$, and $\OPT_k$ is its
minimum. A partition $P$ is \emph{compatible} with a cut $S$ if each of its parts lies inside $S$ or
inside $\Sbar=V\setminus S$. For a global minimum-weight cut $S$ let
\[
  R_k(S) \;=\; \lambda + \min_{1\le j\le k-1}\Bigl(\OPT_j\bigl(G[S]\bigr)+\OPT_{k-j}\bigl(G[\Sbar]\bigr)\Bigr),
\]
with $\OPT_1=0$ and $\OPT_j(H)=\infty$ if $H$ has fewer than $j$ vertices. This is the weight of a
cheapest $k$-partition compatible with $S$, and $R_3(S)=\greedy(S)$.

Compatibility is stronger than Xiao's notion of a cut being \emph{non-crossing} with a
$k$-partition~\cite[Definition~2]{Xiao2008}, which asks only that one side of the cut lie inside a
part. Xiao showed that every global minimum-weight cut is non-crossing with some minimum-weight
3-partition, and that this fails for $k\ge4$~\cite[Lemmas~5 and~9]{Xiao2008}. Compatibility can fail
already for $k=3$: since $R_3(S)=\greedy(S)$, it fails exactly when greedy is not exact, as for the
graph of Theorem~\ref{thm:cex}. A small example for $k=4$: take a 4-cycle with unit weights and a
pendant vertex $p$ attached by an edge of weight $2$. Then $\lambda=2$, attained by $\{p\}$, and
$\OPT_4=4$, by cutting the four cycle edges. But a 4-partition compatible with $\{p\}$ has $\{p\}$ as a
part and cuts the cycle into three pieces, so $R_4(\{p\})=5>\OPT_4$.

\begin{theorem}\label{thm:kcut}
Let $k\ge3$, let $P$ be a $k$-partition with parts $A_i\ne A_j$, and let $S$ be a global minimum-weight
cut. Then
\[
  R_k(S)\;\le\;\max\bigl(W(P)-w(A_i,A_j)+\lambda,\;W(P)\bigr).
\]
In particular, if some minimum-weight $k$-partition has two parts joined by weight at least $\lambda$,
then $R_k(S)=\OPT_k$ for every global minimum-weight cut $S$; that is, every global minimum-weight cut
is compatible with some minimum-weight $k$-partition.
\end{theorem}

\begin{proof}
Merge $A_i$ and $A_j$ to obtain a $(k-1)$-partition $P'$ with $W(P')=W(P)-w(A_i,A_j)$, and refine $P'$
by $S$, replacing each part $X$ by its nonempty pieces among $X\cap S$ and $X\cap\Sbar$. Every edge
crossing the refinement crosses $P'$ or $S$, so the refinement weighs at most $W(P')+\lambda$.

If $S$ splits some part of $P'$, the refinement has at least $k$ pieces. Merging pieces that lie on
the same side of $S$ keeps every piece inside $S$ or $\Sbar$, does not increase the weight, and can
bring the number of pieces down to exactly $k$. This gives a compatible $k$-partition of weight at most
$W(P)-w(A_i,A_j)+\lambda$.

Otherwise every part of $P'$ lies inside $S$ or $\Sbar$. In particular $A_i\cup A_j$ does, so $P$
itself is compatible with $S$, and $R_k(S)\le W(P)$. The second statement follows by taking $P$ optimal
with $w(A_i,A_j)\ge\lambda$.
\end{proof}

Refining a partition by a cut and coarsening the result is a standard device in the analysis of
$k$-cuts; see, for instance, Lemma~5 of Vaughn~\cite{Vaughn2026} and the uncrossing arguments of
Kamidoi, Yoshida and Nagamochi~\cite{KamidoiYoshidaNagamochi2007} and Xiao~\cite{Xiao2008kway}.
What Theorem~\ref{thm:kcut} adds is the observation that a single pair of parts joined by weight at
least $\lambda$ pays for the refinement.

For $k=3$ and the pair $(B,C)$, $W(P)-w_{BC}+\lambda=\lambda+d(A)$, which is
Theorem~\ref{thm:main}.

\begin{corollary}\label{cor:kapprox}
For every $k\ge3$ and every global minimum-weight cut $S$,
\[
  R_k(S)\;\le\;\Bigl(1+\frac{2(k-2)}{k(k-1)}\Bigr)\OPT_k .
\]
\end{corollary}

\begin{proof}
Let $P$ be optimal and $(A_i,A_j)$ its heaviest pair, so $w(A_i,A_j)\ge\OPT_k/\binom k2$. Also
$\lambda\le\min_i d(A_i)\le 2\,\OPT_k/k$, because $\sum_i d(A_i)=2\,\OPT_k$. Theorem~\ref{thm:kcut} gives
$R_k(S)\le\OPT_k+\max\bigl(0,\,2\OPT_k/k-2\OPT_k/(k(k-1))\bigr)$.
\end{proof}

The factor is $4/3$ for $k=3$ and $k=4$, $13/10$ for $k=5$, and tends to $1$. It concerns the best
compatible partition rather than an algorithm, since $R_k$ requires exact solutions on the two sides.

\begin{corollary}\label{cor:khard}
If $R_k(S)>\OPT_k$ for some global minimum-weight cut $S$, then every minimum-weight $k$-partition has
all $\binom k2$ pair weights below $\lambda$, and $k\lambda/2\le\OPT_k<\binom k2\lambda$.
\end{corollary}

\begin{proof}
The first claim is the contrapositive of Theorem~\ref{thm:kcut}. Then
$\OPT_k=\tfrac12\sum_i d(A_i)\ge k\lambda/2$, and $\OPT_k$, a sum of $\binom k2$ pair weights each below
$\lambda$, is below $\binom k2\lambda$.
\end{proof}

In the 5-vertex example above, the pair weights of the optimum are $1$ or $0$, all below $\lambda=2$,
as Corollary~\ref{cor:khard} requires.

\section{Computational checks}\label{sec:comp}

The statements above were tested against brute-force optima on random graphs: planar graphs
(Delaunay triangulations of random points, each edge kept with probability $0.8$) and $G(n,0.45)$
graphs with 5--11 vertices, with weights uniform in $[1,10]$ or uniform integers $1,\dots,5$; and, for
$F(T)$, near-complete, circulant, wheel and three-cluster graphs with 7--10 vertices. The adversarial
search hill-climbs on the weights of random graphs with 11--13 vertices to maximize $F(T)$ minus the
least weight of a 3-partition with no Gomory--Hu side as a part.
All quantities were computed by exhaustive enumeration over vertex subsets and partitions, and every
Gomory--Hu tree was rebuilt from exhaustively computed minimum-weight cuts and re-verified. Each row
is reproduced by one of the supplementary scripts.
\begin{center}
\begin{tabular}{p{0.6\textwidth}rr}
\toprule
Statement & Instances & Violations\\
\midrule
Theorem~\ref{thm:main}, every 3-partition and global minimum-weight cut & 1{,}600 graphs & 0\\
Corollary~\ref{cor:43} ($4/3$ bound), at 1{,}800 global minimum-weight cuts & 1{,}600 graphs & 0\\
Corollary~\ref{cor:hs}: if greedy is inexact, every optimum is hard & 121 graphs & 0\\
Proposition~\ref{prop:hard}, every hard optimum & 429 optima & 0\\
Lemma~\ref{lem:merge-gh}, non-trivial merges in 20{,}000 graphs & 49 merges & 0\\
$F(T)=\OPT$, random graphs & 800 graphs & 0\\
$F(T)=\OPT$, structured graphs & 1{,}577 graphs & 0\\
$F(T)=\OPT$, adversarial weight search (hill climbing) & 240 climbs & 0\\
$F(T)=\OPT$, perturbed optima of the caterpillar program & 72 graphs & 4\\
Theorem~\ref{thm:kcut} condition implies $R_k=\OPT_k$, $k=3,4,5$ & 685 graphs & 0\\
Theorem~\ref{thm:kcut} bound and Corollary~\ref{cor:kapprox}, $k=3,4,5$ & 963 cuts & 0\\
Weaker condition (pairs other than one part), $k=4,5$ & 600 graphs & 24\\
\bottomrule
\end{tabular}
\end{center}
Violations of $F(T)=\OPT$ are thus rare and structured: none appeared in random or locally optimized
instances, only in instances built from the caterpillar program of Section~\ref{sec:cex}. The last
row is expected to fail: it shows that the hypothesis of Theorem~\ref{thm:kcut} cannot be weakened
to ``the parts other than one are joined among themselves by weight at least $\lambda$''.

\paragraph{A numerical caveat.}
Gomory--Hu trees must be built with care in floating point. In early experiments a widely used library
routine for minimum $s$--$t$ cuts occasionally returned a source side inconsistent with the flow value,
because residual capacities of order $10^{-15}$ were treated as usable; the resulting trees were
flow-equivalent but not cut-equivalent, and an adversarial search exploited them to produce false
counterexamples. The supplementary scripts therefore avoid max-flow routines altogether and verify
every tree they use.

\section{Open questions}

\begin{enumerate}[leftmargin=*]
\item \emph{A small candidate family for hard instances.} Nagamochi and Ibaraki~\cite{NagamochiIbaraki2000}
  prove that some optimum has a part among $O(n)$ cuts listed in nondecreasing order up to the first
  crossing pair. Can an equally small family be read off a Gomory--Hu tree alone? Theorem~\ref{thm:cex}
  rules out its sides. Also trying as a part, for every tree vertex, the set of its leaf neighbours
  closes the gap in one of the four failures of Section~\ref{sec:comp} but not in the other three.
\item \emph{Hard instances in near-linear time.} Outside the hard class, weighted minimum 3-cut is
  near-linear on planar and bounded-genus graphs (Corollary~\ref{cor:time}); in general graphs Vaughn's
  $\tilde O(n^2)$ is conditionally optimal. In a plane graph $G$ with dual $H$,
  $\OPT(G)=\min(\greedy,\mathrm{MinTheta}(H))$, where $\mathrm{MinTheta}$ is the least weight of a
  theta subgraph, since every plane theta is 3-separating and non-theta optima are found by greedy
  (Corollary~\ref{cor:hs}(a)). So the question asks for a theta-analogue of the near-linear planar
  girth algorithms. Can the minimum light theta of a plane graph, one whose three strands are each
  lighter than the girth of $H$, be found in $o(n^2)$ time?
\item \emph{Larger $k$, algorithmically.} Theorem~\ref{thm:kcut} settles when a global minimum-weight
  cut is compatible with an optimal $k$-partition, but exploiting this needs exact $j$-cuts of the two
  sides, for which the condition need not hold again. Is there a condition under which the whole
  recursion, with a global minimum-weight cut at every level, is exact?
\item \emph{Hochbaum--Tsai.} Their manuscript~\cite{HochbaumTsai1983} proving the $4/3$ bound is cited
  in~\cite{HochbaumShmoys1985} but appears unavailable; it may contain parts of Theorem~\ref{thm:main}.
\end{enumerate}

\section*{Supplementary material}
The supplementary material (also the ancillary files of the arXiv version, and the public repository
\url{https://github.com/armandorodado/minimum-3cut}) contains:
(i) a Lean~4 proof of every claim of Theorem~\ref{thm:cex}, checked by kernel evaluation with no
additional axioms, together with a script that checks any 7-vertex instance for the properties of
Theorem~\ref{thm:cex} and generates the corresponding Lean proof;
(ii) a Lean~4 formalization, for all $n$, of Theorem~\ref{thm:main}, Corollaries~\ref{cor:hs}
and~\ref{cor:43}, Proposition~\ref{prop:hard}, Lemmas~\ref{lem:merge-cut} and~\ref{lem:merge-gh}, and
the 5-vertex example of Section~\ref{sec:kcut}. It uses natural-number edge weights, which covers
rational weights by scaling, needs no symmetry of the weights, and uses only Lean's standard axioms;
and
(iii) self-contained Python scripts that reproduce every row of the table in Section~\ref{sec:comp},
the linear program of Section~\ref{sec:cex}, an independent exhaustive check of
Theorem~\ref{thm:cex}, and the example of Section~\ref{sec:kcut}.

\section*{Data availability}
No data were used. The code that supports the findings of this article is provided as supplementary
material and at \url{https://github.com/armandorodado/minimum-3cut}.

\section*{Declaration of competing interest}
The author declares that he has no known competing financial interests or personal relationships
that could have appeared to influence the work reported in this paper.

\section*{Declaration of generative AI and AI-assisted technologies in the manuscript preparation process}
During the preparation of this work the author used Claude (Anthropic) to assist with developing and
checking proofs, writing and running the verification code and the Lean formalizations, searching the
literature, and drafting the text. After using this tool, the author reviewed and edited the content
as needed and takes full responsibility for the content of the published article.

\bibliographystyle{elsarticle-num}
\bibliography{greedy3cut}

\end{document}